\documentclass{article}
\usepackage{graphicx} 

\usepackage{amsfonts,amssymb,amsthm,amsmath}
\usepackage{mathabx}
\usepackage{epsfig}
\usepackage{graphicx}
\usepackage{mathrsfs}
\usepackage{color}

\newtheorem{theorem}{Theorem}
\newtheorem{proposition}[theorem]{Proposition}
\newtheorem{lemma}[theorem]{Lemma}
\newtheorem{corollary}[theorem]{Corollary}

\newcommand{\ba}[1]{\begin{array}{#1}}
\newcommand{\ea}{\end{array}}

\newcommand{\ee}{\end{equation}}
\newcommand{\bea}{\begin{eqnarray}}
\newcommand{\eea}{\end{eqnarray}}
\newcommand{\beann}{\begin{eqnarray*}}
\newcommand{\eeann}{\end{eqnarray*}}
\newtheorem{definition}[theorem]{Definition}
\newtheorem{remark}[theorem]{Remark}

\newcommand{\R}{\mathbb{R}}

\newcommand{\nc}{\newcommand}

\nc{\G}{\Gamma}
\nc{\g}{\gamma}
\nc{\al}{\alpha}
\nc{\be}{\beta}
\nc{\del}{\delta}
\nc{\io}{\iota}
\nc{\ka}{\kappa}
\nc{\lam}{\lambda}
\nc{\Lam}{\Lambda}
\nc{\w}{\omega}
\nc{\om}{\omega}
\nc{\Om}{\Omega}
\nc{\Oms}{\Omega^*}
\nc{\s}{\sigma}
\nc{\Si}{\Sigma}
\nc{\ta}{\tau}
\nc{\h}{\theta}
\nc{\z}{\zeta}

\nc{\ut}{t} 

\nc{\ran}{\rangle}
\nc{\lan}{\langle}

\newcommand{\im}{\operatorname{Im}}

\newcommand{\diag}{\operatorname{diag}}

\nc{\bfone}{{\bf 1}}

\newcommand{\n}{\nabla}
\newcommand{\p}{\partial}

\title{Yang-Mills-Higgs Equations on Spherical Orbifolds}
\author{Nicholas M. Ercolani\\
\textit{Department of Mathematics, University of Arizona}\\
\textit{Email: ercolani@math.arizona.edu}}

\date{April 2024}

\begin{document}

\maketitle

ABSTRACT: In this paper we study ({\it static}) solutions of the rank 2 Yang-Mills-Higgs equations on the Riemann sphere, with concical singularities, that  bifurcate from constant curvature connections. We focus attention on the case where there are exactly four such singularities. This study brings together ideas from the gauge theory of constant curvature connections on vector bundles over singular Riemann surfaces with the Riemann-Hilbert analysis  of  classical Fuchsian ODEs. 

\section{Introduction} \label{sec:intro}

Superconductivity is a remarkable instance of situations in which the peculiarities of quantum mechanics manifest themselves on the macroscopic level.  The Ginzburg-Landau equations (GLE) are the Euler-Lagrange equations of a variational model for superconductivity that captures well the qualitative features of resistance drop  
below a critical temperature and magnetic flux expulsion (Meissner effect) \cite{bib:feyn}. The GLE free energy in a 2D model for equilibrium states of superconductors is given by

\begin{eqnarray} \label{GLEen}
\mathbb{E}(\psi, a) &=& \int \left( |\nabla_a \psi |^2 + |\mbox{curl}\,\, a|^2 + \frac{\kappa^2}{2} (|\psi|^2 - 1)^2\right) dxdy \\ \nonumber
(\psi, a)&:& \mathbb{R}^2 \to \mathbb{C} \times \mathbb{R}, \nabla_a = \nabla - i a, - \Delta_a = \nabla_a^* \nabla_a
\end{eqnarray}
where  $\psi$  is a wavefunction  (complex scalar field) for which $|\psi|^2$ represents the number density of Bardeen-Cooper-Schrieffer electron pairs and $a$ provides the vector potential for the magnetic field.  This energy, and corresponding variational equations, are invariant under gauge transformations $g: \mathbb{R}^2 \to U(1),$
\begin{eqnarray*} 
 (\psi, a) \mapsto (g \psi, gag^{-1} + g^{-1} d g).
\end{eqnarray*}
Motivated by the prediction and  discovery of periodic structures of vortex flux tubes in type II superconductors (Abrikosov lattices) \cite{bib:abri, bib:cjrf, bib:et}, it was natural to try to determine  existence of solutions of the GLE equations that are periodic with respect to a  lattice structure \cite{bib:odeh, bib:bgt, bib:tak} and their comparative energetic stability \cite{bib:ts}. Phrased differently, in these works one is studying GLE on Riemann surfaces of genus 1 (elliptic curves). Given all this it was also natural to extend these lines of investigation to higher genus (non-abelian) Riemann surfaces \cite{bib:cers}, 
including non-compact surfaces \cite{bib:esz}.

Another non-abelian extension is to consider these types of models for non-abelian gauge theories, referred to as Yang-Mills-Higgs (YMH) models. This type of generalization presents a number of challenges. The purpose 
of this note is to consider YMH  equations on the Riemann sphere with conical singularities (also referred to as a {\it spherical orbifold}). Introducing singularities might inititally appear to be an unwonted complication; however,  it greatly simplifies the underlying geometry and allows one to more properly focus on the non-abelian aspects of the gauge theory. Indeed, significant literature in the setting of algebraic geometry, has focused upon meromorphic connections on the ``punctured sphere" \cite{bib:al, bib:biq1, bib:bau, bib:lor}.  These investigations build on classical work going back to Fuchs. Fuchsian  differential equations are usually taken to be meromorphic ODEs of the form  
\begin{eqnarray} \label{ODEFuchs}
    \frac{d\Psi}{dz} =  
    \sum_{k  = 1}^n \frac{A_k}{z - z_k}  \Psi(z)
\end{eqnarray}
on the complex plane for fixed distinct points, $z_1, \dots, z_n$, with $\Psi \in \mathbb{C}^m$ and  with $A_k$ being fixed, constant $m \times m$ matrices. This naturally extends to an ODE with regular singular points on the Riemann sphere ( including a singularity at infinity). Global analytic continuation of the fundamental solution of \eqref{ODEFuchs} on the punctured sphere yields a representation of the homotopy group of this surface referred to as the {\it monodromy representation}.  In his $21^{st}$ problem, Hilbert asked if one could associate to any representation of the homotopy group of the punctured sphere  a Fuchsian ODE having that representation as its monodromy represenetation. A related classical question is to determine deformations of the points $\{ z_1, \dots z_n\}$ that leave the monodromy unchanged. This is referred to as the {\it isomonodromy problem} which, in the case that $n=4$, has direct connections to  the Painlev\'e VI equation. We refer the reader to \cite{bib:both} for an informative and up-to-date history of these questions and their status. We also mention that a significant motivation for the modern work concerning connections on punctured spheres came from Hitchen's pioneering application of isomonodromy deformations of meromorphic connections on punctured spheres to the construction of anti-self-dual Einstein metrics on four dimensional space-times  \cite{bib:hitch}.

As will be seen in this article, Fuchsian ODEs provide examples of the extension of the scalar covariant derivative, $\nabla_a$ in \eqref{GLEen}, to the vector setting of YMH. In fact one of the main results of this note (in Section \ref{sec:Holo}) is that any YMH covariant derivative on the spherical orbifold may be gauge reduced to one of Fuchsian type. 

In \cite{bib:dikz}, Deift, Its, Kapaev and Zhou derived, using Riemann-Hilbert methods, explicit solutions of \eqref{ODEFuchs} on the punctured sphere, with further extensions developed in \cite{bib:kor, bib:eg}. As our second result we will see that these explicit constructions (expanded on in Section \ref{sec:Split}) provide a key step in the bifurcation analysis showing the  emergence of lattice solutions for YMH (Corollary \ref{cor:main}). While this provides the mechanism for the explicit construuctions we need, it  requires us to extend the formulation of our gauge theories to handle the inclusion of singularities into the underlying analysis. Fortunately the technology now exists to be able to handle this. This entails our final set of results which are developed in Sections \ref{sec:HOrb} - \ref{OrbConn}.
\bigskip

This study should be viewed as a "proof of concept" foray that hopefully leads to a more general treatment of non-abelian lattice solutions to YMH equations. 
In the conclusions of this note, some potential extensions will be discussed. 
\bigskip

\subsection{The YMH Equations}
\label{YMH}

To define the YMH equations on Riemann surfaces we need to briefly recall some standard constructions from the theory of connections and vector bundles. 
For more detailed background we refer the reader to \cite{bib:wells, bib:koby}. 

Let $\mathcal{R}$ be a Riemann surface with fixed metric $h =  e^\phi dz \otimes d\bar{z}$, and
let $E$ be a vector bundle over $\mathcal{R}$ with fibers isomorphic to a finite-dimensional complex vector space $V$ and a Lie (matrix) group acting on $V$.  
In this paper we will primarily restrict attention to the case where $V$ is two dimensional ($E$ has rank 2) or one dimensional (rank 1, more commonly referred to as a {\it line bundle}). The Lie group will, respectively, be taken to be $U(2)$ or $U(1)$. 

A connection $A$ is a (respectively 2 or 1 dimensional) matrix valued one-form satisfying 
$A^* = -A$ (the Lie algebra of the unitary groups). The connection is associated to a covariant derivative, $\nabla_A$, that maps sections of $E$ (locally vector-valued functions on $\mathcal{R}$) to 1-forms. 
Explicitly, given a section $\Psi$,
\begin{eqnarray*}
   \nabla_A \Psi = \nabla \Psi + A\Psi 
\end{eqnarray*}
 where $\nabla$ is the gradient operator, determined by the metric $h$, and $A$ acts on $\Psi$ by matrix multiplication. Given a curve $\gamma$ on $\mathcal{R}$,  $\nabla_A$ acts on the tangent vector field to the curve by directional derivative with respect to $\nabla$ and by  evaluating the one forms in $A$ on the vector field. This then defines a notion of parallel translation along the curve by solving $\nabla_A \Psi = 0$ along $\gamma$. 
 
 We also assume that $E$  is equipped with a smooth, fiberwise hermitian inner product, $\langle \cdot, \cdot \rangle$. This enables one to define an inner product on sections of the bundle $E$ which extends naturally to sections of the bundle of differential forms with values in $E$: 
 \begin{eqnarray} \label{innerprod}
  ||\eta||^2 = \int_\mathcal{R} \langle \eta \wedge  * \eta  \rangle    
 \end{eqnarray}
where $*$ is the Hodge star operator associated to the metric $h$ on $\mathcal{R}$
\cite{bib:wells}.  This will be defined explicitly  in the cases we need, but  in all cases $\langle \eta \wedge  * \eta  \rangle$ is a real-valued, positive function  multiple of the area form on $\mathcal{R}$. 

From a connection and inner product on $E$, one may define other operators, such as the covariant Laplacian, 
$ -\Delta_{A_{}}: = \n_{A_{}}^*  \n_{A_{}}$, where $\n_{A_{}}^*$ is the adjoint operator with respect to the metrics on $E$ and $\mathcal{R}$. Also, for $B$ a  matrix valued one form, one has 
$d_A B:=d B + \frac12 [A, B]$ \cite{bib:wells}. 
This extends to higher order matrix valued forms in the same way as the ordinary exterior derivative, $d$, does. $d_A^*$ is its adjoint, explicitly given by $d_A^* = -1^k * d_A *$ on $k$-forms. 

The (static) YMH equations for a section $\Psi$ of $E$ and a connection, $A$, on $E$  are 
  \begin{subequations}\label{YMH-eqs}
    \begin{align}  \label{YMH-eqs-psi} 
     & -\Delta_{A_{}}   \Psi  =  \frac{\kappa^2}{2} (1 - |\Psi|^2)\Psi,
    \end{align}   \begin{align}
    \label{YMH-eqs-A}
      &     d_{A}^* d_{A} A = J(\Psi, A), 
   \end{align}
\end{subequations}
where $J(\Psi, A)$ is the YMH current 
given by  $J(\Psi, A):=2  \im (\bar\Psi \otimes  d_{A}\Psi)$, where $(\xi\otimes\eta)_{ij}=\xi_i\eta_j$.
Equations \eqref{YMH-eqs}  are the Euler-Lagrange equation of the YMH energy functional \begin{align}\label{YMHen}
	\mathbb{E} (\Psi, A, h) &= \|\nabla_A \Psi\|^2  + \|d_A A\|^2 + \frac{\kappa^2}{2}\|(|\Psi|^2-1)\|^2. 
\end{align} 

\subsection{Outline}
In the previous subsection we have defined the YMH equations on a smooth Riemann  surface. In the remainder of this paper we want to define how these equations may be extended to compact surfaces with conical singularities, also known as orbifolds. We will focus on the Riemann sphere with such singularities,
spherical orbifolds, and even more particularly the case with just four distinct singularities. In  section \ref{TopCurv} we review the topological and geometric invariants that determine the character of solutions to the YMH equations, again in the case of a smooth Riemann surface.  In section \ref{sec:HOrb} we develop the necessary topological, geometric and analytical background to define the YMH equations on a spherical orbifold. In  particular this section explains the important role played by a smooth elliptic cover of the orbifold. Section \ref{OrbConn} presents a homotopy theoretic alternative description of orbifold connections and goes into the details of this for our case. We also use this to describe the moduli spaces for such solutions and the associated notions of degree and Chern-Weil formula in this singular case. The section concludes with a precise definition of YMH orbifold-solutions and their analytical properties in the case of connections with finite total curvature. Section \ref{sec:Holo} describes the holomorphization of orbifold connections associated to the solution of a d-bar problem. This is where one sees the reduction of general covariant derivatives to Fuchsian ODEs. From that vantage point the analysis of parallel translation with respect to these derivatives can be attacked by known methods of Riemann-Hilbert analysis.
The background for that analysis is developed in Section  \ref{sec:Split}.  In Section \ref{LinSpec} we discuss the linearization of our equations and their associated spectral theory. Finally in  Section \ref{Existence} we outline the proof for the existence and uniqueness of a bifurcating branch of solutions with non-trivial Higgs field. 
\begin{remark}
We note that there are interesting studies of YMH equations on orbifolds in terms of what are known as Higgs bundles \cite{bib:ns}. These are solutions for a particular value of  $\kappa$ where these equations become self-dual. What we do here is a bifurcation analysis with respect to varying values of $\kappa$ and the  solutions we focus on are very far from being self-dual.
\end{remark}

\section{Topology and Curvature} \label{TopCurv}

The description of solutions to \eqref{YMH-eqs} is naturally phrased in terms of certain invariants of $\mathcal{R}$ and $E$.  The first of these is the degree of $E$, denoted $\deg E$.
This is an integer-valued topological invariant. In the case of a line bundle $L$ on $\mathcal{R}$ this can be defined via the short exact sheaf sequence
\begin{eqnarray*}
    0 \to \mathbb{Z} \to \mathcal{O} \to \mathcal{O}^* \to 0
\end{eqnarray*}
where $\mathcal{O}$ is the sheaf of holomorphic functions on $\mathcal{R}$ and $\mathcal{O}^*$ is the sheaf of non-vanishing holomorphic functions on $\mathcal{R}$ with the map between them given by $\exp 2\pi i (\cdot)$.  This induces, in the associated long exact cohomology sequence, a connecting homopmorphism,
\begin{eqnarray*}
    H^1(\mathcal{R}, \mathcal{O}^*) \to H^2(\mathcal{R}, \mathbb{Z}) = \mathbb{Z}.
\end{eqnarray*}
The elements of $H^1(\mathcal{R},\mathcal{O}^*)$ are $\check{C}ech$ equivalence classes \cite{bib:wells} of transition functions on local trivializations of line bundles and so are in 1:1 correspondence with topologically inequivalent line bundles on $\mathcal{R}$. The above homomorphism defines the degree of $L$, $\deg L$. More generally, for bundles, $E$, of any rank $m$,  this degree is given by $\deg E := \deg \Lambda^m E$ where $\Lambda^m E$ is the determinant line bundle of $E$. (In all cases this degree is also referred to as the first Chern
class, $c_1(E)$, of the bundle.) 

The second important fact here is that YMH is a gauge theory. This means that for $g(z) \in C^1(X, U(m))$, if $(\Psi, A)$ solve \eqref{YMH-eqs} on $E$, then 
\begin{eqnarray*}
    T_g(\Psi,A) := \left(g\Psi, g^{-1}Ag + g^{-1} dg   \right)
\end{eqnarray*}
is another section-connection pair that solves \eqref{YMH-eqs} on $E$.

Finally we introduce the {\it curvature of $A$}, defined by
$F_{A}:= d_A A$. This is a matrix-valued two-form that transforms by simple conjugation under the gauge action.
  A connection $A$ on $E$ is said to be a {\it constant curvature connection} if its  curvature is of the form 
 \begin{align} \label{const-curv-cond} F_{A_{}}= -i b \mathbb{I} \otimes \om,   
 \end{align} 
 for some $b\in \R$, where $\om$ is the area form on $\mathcal{R}$. When $b=0$ the connection is said to be {\it flat}.

 The topological invariant, $\deg E$, may be analytically expressed in terms of the curvature through the 
 Chern-Weil relation
 \begin{eqnarray} \label{CW}
     \deg E = -\frac{1}{2 \pi i} \int_\mathcal{R} \mbox{Tr} F_A.
 \end{eqnarray}
In the constant curvature case this determines the constant $b$,
\begin{eqnarray} \label{magflux}
    b = \frac{2 \pi \deg E}{|\mathcal{R}| m(E)},
\end{eqnarray}
where $|\mathcal{R}|$ denotes the total area of the surface and $m(E)$ is the rank of $E$. 
We note from this relation that, for fixed metric and hence fixed total area, the constant $b$, which corresponds to the total magnetic flux of the YMH  solution, is quantized.

\subsection{Normal Solutions: Yang-Mills Connections}

When $\Psi \equiv 0$ one refers to these solutions of  \eqref{YMH-eqs} as being {\it normal}. They correspond to stationary solutions of the Yang-Mills functional, 
$\mathbb{E}_{YM}(A) = ||d_A A||^2 = ||F_A||^2$.
The condition for stationarity is
\begin{eqnarray*}
    d_A^* d_A A = 0.
\end{eqnarray*}
Noting that 
$d_A^* d_A A = d_A^* F_A = *d_A*F_A$ and combining this with Bianchi's identity \cite{bib:wells}, $d_A F_A = 0$, shows that being 
stationary amounts to a nonlinear analogue of the condition for a two-form to be harmonic (closed and co-closed). A flat connection is clearly stationary; but, more generally, so is any constant curvature connection. This follows from the definition $d_A B := dB + \frac12 [A,B]$ since, in the constant  curvature case, $*F_A = -i b \mathbb{I}$.

\section{Orbifolds} \label{sec:HOrb}
 
Our goal in this paper is to make a step towards extending the GLE results of \cite{bib:cers, bib:esz} to YMH. We will describe this here just for the case of a Riemann sphere with conical singularities. As mentioned  in the Introduction,
this is a class that has been studied fairly recently in the literature related to moduli theory of bundles and Yang-Mills theory. The origins for this stem from  the work of Seshadri and collaborators   \cite{bib:sesh, bib:ms}. 

We refer to \cite{bib:troy} for a detailed background on what a general orbifold is; but,  in this section, we present an explicit example of an orbifold structure on the sphere that will exhibit the essential idea  of this structure as well as suffice for all of the purposes of this paper. However,  first we need to introduce some relevant definitions and notations.

\subsection{Spherical Orbifold Structure Induced by an Elliptic Cover }
\label{inducedElliptic}

We begin with a Riemann  surface of genus 1 (often referred to as an {\it elliptic curve}) that is determined by the equation

\begin{eqnarray} \label{defeqn}
    y^2 = \prod_{i=1}^3 (z - z_i), 
  \end{eqnarray}
for complex variables $(z,y)$ and where $z_1, z_2$ and $z_3$ are fixed complex numbers as described in the Introduction. There is a natural involution on this surface given by 

\begin{eqnarray}  \label{invol}
    \iota:  (z,y) &\to & (z, -y).
\end{eqnarray}
with fixed points at $w_i = (z_i, 0)$.
This surface has a one point compactification  by  adding a point at infinity which is another fixed point of
of $\iota$ and we'll denote it by  $w_4 = (\infty, 0)$. In  this way the compact surface,  which going forward we'll denote by $C$, presents as a branched double cover of $\mathbb{S}$ where
\[
\mathbb{S} = \mathbb{C}  \cup \{ \infty \},
\]
the {\it Riemann  sphere}. The points $w_1, \dots, w_4 \in C$ are referred to as the Weierstrass points and the corresponding $z_1, \dots, z_4 (=\infty) \in \mathbb{S}$ are referred to as branch points.   We let $\Sigma$ denote 
$\mathbb{S} \backslash \{ z_1, \dots, z_4\}$ and 
$\hat{\Sigma} = C \backslash \{w_1, \dots, w_4\}$. 
$\hat{\Sigma}$  is an (unbranched) double cover of $\Sigma$. The branched covering map is the topological quotient of $C$ by the action of the order 2 group, $G$,
generated by \eqref{invol}. Given a point $z \in \mathbb{S}$ let $P^{(i)}, i  = 1,2$ denote the two points in the corresponding equivalence class of the quotient, $(z, \pm y)$:
\begin{eqnarray} \label{quotient1}
    C \to \mathbb{S}  &=& C/G\\ \nonumber
     P^{(i)} \mapsto z &=&  \left\{ P^{(1)}, P^{(2)}\right\}.
\end{eqnarray}
 Analytically we can now define a holomorphic differential on $C$. Such a differential, up to a constant multiple, is uniquely given by
 \begin{eqnarray}
   \phi = dz/y.    
 \end{eqnarray}
We fix a basepoint, $z_0 \in  \mathbb{S}, z_0 \ne  z_i$, and  some choice, $P_0 \in C$ over $z_0$. Then we define the {\it Abel map} 
\begin{eqnarray*}
    \mathcal{A}: C  &\to& \mathbb{C}\\
    P &\mapsto&  \int_{P_0}^P \phi.
\end{eqnarray*}
Let $\{a,b\}$ be a basis of the homology group $H_1(C, \mathbb{Z})$. Then
$\omega_1 = \oint_a \phi, \omega_2 = \oint_b \phi$  generate the period lattice $\Gamma$,  of  $\phi$, in $\mathbb{C}$. These periods may be normalised to the basis
$\{ 1, \tau = \omega_2/ \omega_1\}$.  $\mathcal{A}$ then determines a conformal isomorphism (uniformization)
\begin{eqnarray*}
    C \simeq \mathbb{C}/\Gamma. 
\end{eqnarray*}
For more details on  any of the above background we refer the reader to \cite{bib:fk}.
\medskip

We can now define the orbifold structure on $\mathbb{S}$ that we will use.
Taking  $u$ to be the coordinate on $\mathbb{C}$ we define a Riemannian metric structure on $C$ corresponding to the Euclidean metric on  $\mathbb{C}$, with 
$\Gamma$ realized as lattice translations (isometries) 
\begin{eqnarray*}
(C, d\hat{s}^2) = (\mathbb{C}/\Gamma, |du|^2).
\end{eqnarray*}
Under this conformal isomorphism, the involution $\iota$ on $C$ correponds to 
the involution $u \to -u$. The quotient under this involution in a neighborhood of $u=0$ is metrically a cone of total rotational angle $\pi$. For this reason such  a point is referred to as a {\it coniical singularity} of the orbifold. By periodicity under the translation isometries $\Gamma$, with respect to the basis $\{ 1, \tau \}$, there are isometric singularities at three other points, $u \equiv 1/2,  \tau/2, 1/2 + \tau/2$.  These four {\it half periods} correspond, under the Abel map, to the Weierstrass points , $\{w_1, \dots, w_4\}$ on $C$. 

We further note that the covering map extends to the globally conformal map,
\begin{eqnarray} \label{quotient2}
    \mathbb{C}/\Gamma &\to & \mathbb{S} \\ \nonumber
    u &\mapsto& z = \wp(u),
\end{eqnarray}
where $\wp$ is the Weierstrass $P$-function, with further quotient $G\backslash\mathbb{C}/\Gamma$ corresponding to \eqref{quotient1}. Under  this correspondence the half periods map to the branch points  $\{z_1, \dots, z_4\}$ on $\mathbb{S}$.  In this way $|du|^2$ induces a well defined metric on the punctured sphere $\Sigma$ which  develops conical singularities as one approaches the branch points. This gives $\mathbb{S}$ an orbifold structure. All that will be needed going forward is to use this correspondence to define metrical structures on $\Sigma$ by pulling them back to $C$ and using the Euclidean structure one has there. As a first instance of that we consider the Hodge $*$-operator. For the flat metric on $C$ specified above, one has
\begin{eqnarray} \label{star}
    *du  &=&  -i \,\, du\\ \nonumber
    * d\bar{u}  &=& i \,\, d \bar{u}.
\end{eqnarray}
 Given a covariant derivative $\nabla_A$ or covariant differential $d_A$ on $\Sigma$
 these pull back naturally, with respect to \eqref{quotient2},  to corresponding operators definied on $\hat{\Sigma}$, in which the pullback of $\nabla$ is  the gradient operator with respect to the Euclidean metric $|du|^2$ on $\mathbb{C}$ and $d$ is the standard exterior derivative on $\mathbb{C}$. On $\mathbb{C}/\Gamma$, the $*$ operator is defined as specified in \eqref{star}. Then adjoints with respect to the inner product \eqref{innerprod} are determined as before to be
 \begin{eqnarray} \label{starnab}
     \nabla_A^* &=& -*\nabla_A * \\ \label{stard}
     d_A^* &=& -  * d_A * .
 \end{eqnarray}
When there is no chance of confusion we will speak of all these pulled back operators as being {\it on} $\Sigma$.

\subsection{Bundles and Connections on the Spherical Orbifold}

We want to work in a setting that extends the model of Fuchsian differential equations mentioned in the Introduction. To that end we will start out taking $E$ to be the trivial rank 2 bundle on $\Si $; i.e. $E \simeq \,  \Si \times \mathbb{C}^2$. (This coincides with the trivial bundle over $\mathbb{S}$ restricted to $\Si$.) Since the bundle $E$ is defined locally on $\Si$, it pulls back naturally to a bundle over $\hat{\Si}$ with the same local trivializations on each sheet of the double cover. We denote the pullback of this bundle to $\hat{\Si}$ by $\hat{E}$. This extends naturally to a smooth bundle on $C$ by identifying the fibers of the bundle on the two sheets of $\Si$ as they limit into a Weierstrass point of $C$. We will continue to denote this extended rank 2 bundle by $\hat{E}$. This is also a trivial bundle: 
$\hat{E} \simeq \, C \times \mathbb{C}^2$.

 As described at the end of section \ref{inducedElliptic}, a connection $A$ on $E$ over ${\Si}$ is a matrix valued 1 form and so it naturally pulls back to a connection $\hat{A}$ on $\hat{E}$ over $\hat{\Si}$, with corresponding covariant derivative $\nabla_{\hat{A}}$.  
We will see in section \ref{sec:Holo} that both $A$ and
$\hat{A}$ are respectively gauge equivalent to connections that extend to be well-defined meromorphic connections with regular singular points; i.e., they are equivalent to Fuchsian differential equations.  

\subsection{YMH equations on the Orbifold}

Extending the covariant differential defined at the end of section  \ref{inducedElliptic}, one has the following sequence of operators on form-valued sections \cite{bib:wells}.

\begin{eqnarray} \label{precomplex}
\mathcal{E} \xrightarrow{d_{\hat{A}}} &\Omega^1(\mathcal{E}) & \xrightarrow{d_{\hat{A}}} \Omega^2(\mathcal{E}),
\end{eqnarray}
where $\mathcal{E}$ is the sheaf of  smooth sections of $\hat{E}$ on $\hat{\Si}$ and $\Omega^k(\mathcal{E})$ is the sheaf of smooth $k$-form valued sections. 
For this to be well-defined we need to consider the behavior of $\hat{A}$ in a deleted neighborhood of the points $w_i$ on $C$. For that we make use of the following result.
\begin{theorem} \cite{bib:biq1} \label{biquard}
   Let $A$ be a unitary connection on a holomorphic vector bundle $E$ over $\hat{\Sigma}$ with $L^2$-curvature; then $A$ admits a local holonomy around $w_j$ with eigenvalues $e^{2 \pi i \hat{\lambda}_j^\pm}$ and there exists a gauge in a deleted neighborhood of $w_j$ for which \begin{eqnarray*}
       d_A = d + i \left( \begin{array}{cc}
           \hat{\lambda}_j^+ & 0 \\
           0 & \hat{\lambda}_j^-
       \end{array}\right) d\theta +a
   \end{eqnarray*}  
   with $|| a/r||_{L^2} + || \nabla a||_{L^2} < \infty$ where
   $(r, \theta)$ are local polar coordinates around $w_j$. 
 \end{theorem}
 Since at this stage we are taking $\hat{E}$ to be the trivial bundle on $C$, it is certainly holomorphic.  (Other extensions of this bundle appearing in this paper will also be holomorphic.) Moreover, we will always be working with connections whose curvature is $L^2-bounded$. Hence, in this paper, it follows from the local representation of the theorem that \eqref{precomplex} will always be well-defined. 

 One may now define natural Hermitian inner products on sections of the sheaves in \eqref{precomplex} by 
\begin{eqnarray} \label{norm}
(\xi, \eta) &=& \int_{\hat{\Si}} \xi \wedge *\bar{\eta}.
\end{eqnarray}
For matrix-valued sections we understand the integrand above to mean $\mbox{Tr} \,\, \xi \wedge *\bar{\eta}$. This norm endows the spaces of global sections of these sheaves with Hilbert space structure. With this one may extend \eqref{precomplex} to be defined for $L^2$ sections over $C$. One also has the dual sequence to \eqref{precomplex},

 \begin{eqnarray} \label{dualcomplex}
\mathcal{E} \xleftarrow{d^*_{\hat{A}}} &\Omega^1(\mathcal{E}) & \xleftarrow{d^*_{\hat{A}}} \Omega^2(\mathcal{E}),
\end{eqnarray}
in terms of the adjoint covarint exterior derivative, $d^*_{\hat{A}}$, with explicit form as given in \eqref{stard} in terms of the $*$-operator.
One may then define covariant Laplacians on each sheaf as

\begin{eqnarray*}
    \Delta_{\hat{A}} = d^*_{\hat{A}}d_{\hat{A}} + d_{\hat{A}}d^*_{\hat{A}}.
\end{eqnarray*}
On $\mathcal{E}$ this reduces to just
$d^*_{\hat{A}}d_{\hat{A}}$ which will be the operator of interest to us (which corresponds to the covariant Laplacian appearing in \eqref{YMH-eqs-psi}).

\section{Orbifold Connections and Holonomy Representations}
\label{OrbConn}

As described in Section 1, a unitary connection $A$ on a bundle $E$ defines a covariant derivative $\nabla_A$ on the local sections of $E$.  Locally this amounts to a system of homogeneous ODE's whose solutions define a parallel transport of sections along paths emanating from a base point $z_0 \in \mathbb{S}$. When the curvature of $A$ is zero, these ODE solutions may be globally continued along any closed path on the punctured sphere $\Si$. Parallel transport along a closed path then produces an automorphism of the initial frame of sections based at $z_0$. With respect to this initial frame, one then gets a representation $\rho$ of the fundamental group $\pi$ (denoting the homotopy group $\pi_1(\Si, z_0)$ of closed paths based at $z_0$), referred to as the {\it holonomy} representation.
For the case of rank 2 vector bundles this is a representation in  $U(2)$.  

\subsection{Homotopy Considerations} \label{homotop}

Taking a point $z_0 \ne z_i$ as a base point let $\gamma_i$ denote the homotopy class of a counterclockwise loop based at $z_0$ that encircles $z_i$ on $\mathbb{S}$, but no other marked points, $z_j$. These four $\gamma_i$ generate the homotopy group of $\Si$, $\pi_1(\Sigma, z_0)$,  subject to the  relation $\gamma_1 \cdots \gamma_4 = 1$.

As mentioned in section \ref{inducedElliptic}, $\hat{\Sigma}$ is an unbranched double cover of $\Sigma$. By standard homotopy theory there is a Galois theoretic correspondence between coverings of a space and normal subgroups of the homotopy group of that space. It follows from this that
\begin{eqnarray*}
    \pi_1(\hat{\Sigma}, P^{(i)}_0) \subset \pi(\Sigma, z_0)
\end{eqnarray*}
is a normal subgroup of index 2. Elements of this subgroup can all be expressed as words of even length in the generators $\gamma_1, \gamma_2, \gamma_3$.

The homotopy group of $C$ itself is abelian, isomorphic to its homology group which in turn is isomorphic to the period lattice $\Gamma$ defined in section \ref{inducedElliptic}. As elements of a subgroup of $\pi = \pi_1(\Sigma, z_0)$ these homology cycles may be related to the  monodromy representation $\rho$. The following lemma makes this precise.
\begin{lemma} \cite{bib:hitch} \label{Hitchen}
$\rho(\pi_1(\hat{\Sigma}, P^{(i)}_0)) \simeq \Gamma$. In terms of the homotopy generators, the matrices ($\in SL(2, \mathbb{C})$) correseponding to the even words in this subgroup satisfy the relations: $\rho_i^2 = -1, \rho_1 \cdots \rho_4 = 1$, where $\rho_i = \rho(\gamma_i)$. The image of the full homotopy group under $\rho$ is a semi-direct product $\Gamma \rtimes G$ where $G$ is generated by the elliptic involution $\iota$; i.e., $\rho(\pi)/\Gamma \simeq G$.  
\end{lemma}

\subsection{YMH Solutions on the Orbifold}

We can now define what we mean by the orbifold YMH equations and free energy, generalizing what was defined for the smooth case in Section \ref{YMH}. We consider pairs $(\Psi,\hat{A})$ on $\hat{E}$ for which the free energy \eqref{YMHen} is finite. (The norms in the energy are defined with respect to \eqref{norm}.) From this it clearly follows that $\hat{A}$ has finite curvature and so Theorem \ref{biquard} applies.  The formal Euler-Lagrange equations of this free energy reproduce the YMH equations \eqref{YMH-eqs} with the operators appearing there understood in the sense defined in Section \ref{sec:HOrb}. Finally to define a solution for YMH on $(\Si, E)$ we additionally require that, for the solution $(\Psi, \hat{A})$, $\hat{A}$ is the pull-back of a connection $A$ on $E$ and that $\Psi$ be equivariant with respect to the action of the elliptic involution $\iota$ defined in \eqref{invol}. This means that 
$\iota^* \Psi = \rho(\iota) \Psi$.  

We detail this in the following definitions and proposition.
\begin{definition} \label{YMH-orb}
 A solution to the YMH equations on the orbifold $(\Si, E)$ is defined to be a finite energy critical point $(\Psi, \hat{A})$ on $\hat{\Si}$ such that $\hat{A}$ is the pull-back of a connection $A$ on $\Si$ and 
 $\Psi$ is equivariant with respect to $\iota$. 
\end{definition}

\begin{definition} \label{sobolev}
    We let $\mathcal{H}^s$ denote  the space of Sobolev sections of $\hat{E}$, extended over $C$, of order $s$,
    and $\vec{\mathcal{H}^s}$ denote the Sobolev space matrix-valued sections of  $\Omega^1(\mathcal{E})$ of order $s$.  By {\it extension over} $C$ one means that the bundle is extended over the Weierstrass points, $w_i$, on $C$ so that the fibers of $\hat{E}$ over the two sheets of the elliptic curve coincide in the limit as they  approach $w_i$.
\end{definition}

\begin{proposition} \label{distribution}
 If $(\Psi, \hat{A})$ has finite YMH energy, \eqref{YMHen}, then $\Psi \in \mathcal{H}^1 \cap L^4$  and
$\hat{A}  \in  \vec{\mathcal{H}^1}$. Given  this, the variational equations \eqref{YMH-eqs} are  defined in the distributional sense on $\mathcal{H}^1 \times \vec{\mathcal{H}}^1$.    
\end{proposition}
\begin{proof}
    It follows from the finiteness of the energy  that
    \begin{eqnarray} \label{est1}
        ||(|\Psi|^2 - 1)||^2 &<& C_1^2 < \infty \\ \label{est2}
        ||\nabla_{\hat{A}} \Psi||^2 &<& C_2^2 < \infty\\ \label{est3}
        ||d_{\hat{A}} \hat{A}||^2 &<& \infty
    \end{eqnarray}
    From \eqref{est1} we have 
    \begin{eqnarray*}
     || \, |\Psi|^2||   &=& ||(|\Psi|^2 - 1) + 1||   \\
     &\leq& ||(|\Psi|^2 - 1)|| + || 1 ||\\
     &=& C_1 + \mbox{area}(C) < \infty,
    \end{eqnarray*}
    from which it follows that $\Psi \in L^4$. Since $C$ has finite area, it  follows that $\Psi $ is also in $L^2$. By \eqref{est2}, 
    \begin{eqnarray*}
||\nabla \Psi|| &=&  
||\nabla \Psi +  \hat{A}\Psi - \hat{A}\Psi|| \\
&\leq& ||\nabla \Psi +  \hat{A}\Psi ||  + || \hat{A} \Psi ||\\
&\leq&  ||\nabla_{\hat{A}} \Psi|| + ||\hat{A}|| || \Psi ||\\ 
&\leq  & C_2 + || \hat{A}|| || \Psi || <  \infty
    \end{eqnarray*}
    where  $|| \hat{A}  ||$ is bounded by  Theorem \ref{biquard} and $||\Psi||$
    is bounded as we saw from the previous inequalities. It follows that
    $\Psi \in \mathcal{H}^1 \cap L^4$. Since $d_{\hat{A}} \hat{A} = F_{\hat{A}}$ ,  \eqref{est3}  means  that the curvature is bounded and so, by Theorem \ref{biquard}, $$\hat{A}  \in  \vec{\mathcal{H}}^1.$$ 
\smallskip

    The variation of the energy \eqref{YMHen} at a minimum, $(\Psi, \hat{A})$, 
    yields the equations
    \begin{eqnarray} \label{YMHvarpsi}
        (\nabla_A \Psi, \nabla_A \Phi) &=& \left(\frac{\kappa^2}{2} (1 - |\Psi|^2)\Psi, \Phi  \right)\\ \label{YMHvaralpha}
        (d_A A, d_A \alpha) &=& (J(\Psi, A), \alpha)
    \end{eqnarray}
    for any $\Phi, \alpha$ where the latter are Schwarz class sections.  For classical solutions, equations \eqref{YMH-eqs} are equivalent to the previous equations. For the finite energy sections this variational formulation defines the notion of a distributional solution to equations \eqref{YMH-eqs}.  That equations \eqref{YMHvarpsi} and \eqref{YMHvaralpha} are well-defined follows from the first part of this proposition and several applications of Cauchy-Schwarz and an application of Jensen's and H\"older's inequalities on a compact surface to control the RHS of \eqref{YMHvarpsi}.   
\end{proof}

\subsection{Moduli of Orbifold Connections}

The homotopy considerations presented in section \ref{homotop} lead to an alternative description of connections in terms of the homotopy group of the punctured sphere. Letting $\rho$ denote a representation of $\pi$ on $\mathbb{C}^2$. Then $\pi$ lifts to an action by bundle automorphisms on the trivial $\mathbb{C}^2$ - bundle  
$$
\mathbb{C} \times \mathbb{C}^2 \to \mathbb{C} 
$$
defined by
\begin{eqnarray*}
    \Gamma  \times \mathbb{C} \times \mathbb{C}^2 &\to & \mathbb{C} \times \mathbb{C}^2\\
    \gamma \times (u,v) &\mapsto& 
    (\gamma u , \rho(\gamma) v).
\end{eqnarray*}
This is a proper action with quotient being the bundle
\begin{eqnarray*}
   (\mathbb{C} \times \mathbb{C}^2)/\Gamma  \to \mathbb{C}/\Gamma  = C
\end{eqnarray*}
which we'll denote by $\hat{E}_\rho$. The elliptic involution $(u \to -u)$ factors through this bundle projection to give the bundle
\begin{eqnarray*}
 (\mathbb{C}/G \times \mathbb{C}^2)/\Gamma 
 = (\mathbb{C} \times \mathbb{C}^2)/(\Gamma \rtimes G)
 \to G \backslash\mathbb{C}/\Gamma  = \mathbb{S}   
\end{eqnarray*}
which we denote by $E_\rho$. $\hat{E}_\rho$ is the pullback of the bundle ${E}_\rho$ from $\mathbb{S}$ to $C$.
\bigskip

Given a connection on $\mathbb{S}$ with regular singular points at the branch points $z_j$, let  $\lambda_j^\pm$ be the eigenvalues of the connection at the $z_j$. We consider the collection of all such connections having these same eigenvalues at the branch points. We'll refer to these as  $\lambda_j^\pm$-connections. 
 Based on what has just been described, one may define a map
 \begin{eqnarray*}
     [\mbox{flat  $\lambda_j^\pm$ smooth connections on}\,\, E_{|\Si}] \to Hom(\pi, GL(2, \mathbb{C}))
 \end{eqnarray*}
by mapping such a connection to its holonomy representation.  One can  show \cite{bib:koby, bib:ms} that given any such representation of $\pi$ there is a connection on $E_{|\Si}$, unique up to gauge transformation, whose holonomy realizes that representation. So in fact one has a homeomorphism
\begin{eqnarray} \label{moduli}
     [\mbox{flat  $\lambda_j^\pm$-connections}/\{ \mbox{gauge equivalence}\}] \xrightarrow \sim Hom(\pi, GL(2, \mathbb{C}))/ GL(2, \mathbb{C}).
 \end{eqnarray}
The quotient by $GL(2, \mathbb{C})$ on the right corresponds to the freedom to choose different initial frames at the base point $z_0$.

There are natural global coordinates that can be prescribed for\\  $Hom(\pi, GL(2, \mathbb{C}))/ GL(2, \mathbb{C})$
that were introduced by Seshadri \cite{bib:sesh} in the more general setting of a punctured Riemann surface. The moduli space \eqref{moduli} can be coordinatized as the fiber over the identity, $\mathbb{I}$, in the multiplication map
\begin{eqnarray} \label{fibermoduli}
    C_1 \times C_2 \times C_3 \times C_4 \to GL(2, \mathbb{C})
\end{eqnarray}
where $C_j$ is the conjugacy class of matrices with eigenvalues  $\exp 2 \pi i \lambda^+_j, \exp 2 \pi i \lambda^-_j$.

\subsection{Parabolic Degree}
In the case of a compact Riemann surface, the degree of a bundle supporting a flat connection must be 0. We consider a more general notion of bundle type and degree for a punctured surface such as $\Si$. We are considering trivial $GL(2, \mathbb{C})$ bundles on the sphere but with connections having poles of order 1 at each of the marked points $z_j$ (Fuchsian connections) and eigenvalues $\lambda_j^\pm$ of the matrix residues fixed. As before, we'll refer to these as as $(\lambda_1^\pm; \dots; \lambda_4^\pm)$-bundles. By the homeomorphism \eqref{moduli} these bundles correspond to those $E_\rho$ for which the generators of $\pi$ map to matrices with respective eigenvalues $\lambda_j^+, \lambda_j^-$. 
If the structure group is reduced to $U(2)$, which will be the case for unitary connections, the $\lambda_i^\pm$ are real.

Following \cite{bib:sesh} one may  extend the notion of degree to the orbifold setting.  Given a $(\lambda_1^\pm; \dots; \lambda_4^\pm)$-bundle, $E$, define its {\it parabolic degree} to be
\begin{eqnarray} \label{pardeg}
    \deg \mbox{par} E = \deg E +
    \sum_{j=1}^4 (\lambda_j^+ + \lambda_j^-).
\end{eqnarray}
For this one has an extension of the Chern-Weil  formula \eqref{CW},

\begin{eqnarray} \label{CWOrb}
       -\frac{1}{2 \pi i} \int_\Si \mbox{Tr} F_A
       = \deg \mbox{par} E.
 \end{eqnarray}
(See for instance \cite{bib:biq}.)
The bundles $E_\rho$ constructed earlier have degree 0, i.e. they are flat. 
It then follows from the residue theorem that $\deg \mbox{par} E = 0$ as well. 
\smallskip

One can also build non-flat bundles, $\tilde{E}$, by {\it twisting} the transition functions for 
$ E_\rho$ by that of a non-trivial line bundle $L$ (smooth over $\mathbb{S}$) with a given degree. Then the degree is given by $ \deg \tilde{E} = \deg \left(E_{\rho} \otimes L\right) = 2 \deg L$ with $\deg \mbox{par} \tilde{E}$ then given by \eqref{pardeg}. The moduli space in this case is coordinatized by  the fiber over $(\deg \mbox{par} \tilde{E}) \,\, \mathbb{I}$ in \eqref{fibermoduli}. 
\smallskip

With these definitions in place we can also extend the expression \eqref{magflux} for the total magnetic flux to
\begin{eqnarray*} 
    b = \frac{2 \pi \deg \mbox{par} \tilde{E}}{|\Sigma| m(E)}
\end{eqnarray*}
which in our case becomes
\begin{eqnarray} \label{extmagflux}
    b &=& \frac{2  \pi \cdot 2\deg L}{|\Si| \cdot 2}\\  \nonumber
    &=&  \frac{2  \pi \deg L}{|\Si|}. 
\end{eqnarray}

In what follows we are going to assume that $\lambda_j^\pm = \pm 1/4$. This case allows us to make use of calculations done elsewhere  \cite{bib:hitch, bib:dikz}.

\section{Holomorphization: Holonomy into Monodromy} \label{sec:Holo}

Following our considerations in Section \ref{OrbConn} we may regard the unitary connection
$A$ as a one form taking its values in general $2 \times 2$ matrices over $\mathbb{C}$ rather than restricted to skew Hermitian matrices. With this in mind,  one can decompose it with respect to the complex structure on the Riemann surface as
\begin{eqnarray*}
    A &=& A^{(1,0)} + A^{(0,1)} 
\end{eqnarray*}
where $A^{(1,0)}$ is a matrix multiple of $dz$ and $A^{(0,1)}$ is a matrix multiple of $d\bar{z}$. By extension, the covariant exterior derivative decomposes as
\begin{eqnarray*}
    d_A &=& \partial_A' + \partial_A''\\
   \partial_A' &=& \partial +  A^{(1,0)}\\
   \partial_A'' &=& \bar{\partial} + A^{(0,1)}
\end{eqnarray*}
where $\partial = \frac{\partial}{\partial z} \wedge dz$ and 
$\bar{\partial} = \frac{\partial}{\partial \bar{z}} \wedge d\bar{z}$.

The following fundamental result is a consequence of the Newlander-Nirenberg theorem,

\begin{theorem} \cite{bib:koby} \label{thm:NN}
Given a smooth complex vector bundle $E$ over a complex manifold $M$ with connection $\nabla_A$ such that 
$\partial_A'' \circ \partial_A'' = 0$, there is a unique holomorphic vector bundle structure on $E$ such that $\partial_A'' = \bar{\partial}$. In other words, $A^{(0,1)} \equiv 0$.
\end{theorem}
In the case that $M$ is a Riemann surface, the type decomposition of the curvature tensor is 
\begin{eqnarray*}
    F_A = F_A^{(2,0)} + F_A^{(1,1)}  + F_A^{(0,2)},
\end{eqnarray*}
but   $F_A^{(2,0)} = 0 = F_A^{(0,2)}$
since $dz \wedge dz = 0 = d\bar{z} \wedge d\bar{z}$. Because $F_A^{(0,2)} = \partial_A'' \circ \partial_A''$, the condition of the Theorem is satisfied for any Riemann surface; this includes $\Si$ and $C$ from the previous sections.

If the connection is flat, then
\begin{eqnarray*}
 0 = d_A A &=& (\partial_A' + \partial_A'') (A^{(1,0)} + A^{(0,1)} ) \\
 &=& (\partial_A' + \partial_A'') A^{(1,0)} \\
 &=& \partial_A'' A^{(1,0)}\\
 &=& \bar{\partial} A^{(1,0)}.
\end{eqnarray*}
It follows that the coefficients of 
$A^{(1,0)}$ are holomorphic. In the case of $\Si$, the growth conditions we've required near the $z_i$ implies, as a consequence of the Riemann singularity theorem,
that the coefficients of $A^{(1,0)}$ extend to be meromorphic at the $z_i$ with simple poles. Hence $\partial_A'$ is a classical Fuchsian differential operator. The condition for parallel translation of a section, $\Psi$, around a closed loop $\gamma$, $\nabla_A \Psi = 0$, is now expressed as solving  
\begin{eqnarray} \label{FuchsOp}
  \partial_A' \Psi =  (\partial +  A^{(1,0)}) \Psi = 0.
\end{eqnarray}
Parallel translation is then reduced to the analytic continuation of the solution to the meromorphic, Fuchsian ODE on  $\mathbb{S}$,  

\begin{eqnarray} \label{FuchsEqn}
    \frac{d\Psi}{dz} =  
    \sum_{k  = 1}^3 \frac{A_k}{z - z_k}  \Psi(z)
\end{eqnarray}
where the residue matrices $A_k$ have eigenvalues $\pm 1/4$.
The holonomy of the connection is thus identified with the monodromy of this ODE.

\begin{remark}
In this holomorphic setting the sequence \eqref{precomplex} reduces to the pair of sequences
\begin{eqnarray} \label{complex}
\mathcal{E} \xrightarrow{\partial'_{\hat{A}}} &\Omega^1(\mathcal{E}) & \xrightarrow{\partial'_{\hat{A}}} \Omega^2(\mathcal{E})\\ \label{complex2}
\mathcal{E} \xrightarrow{\partial''} &\Omega^1(\mathcal{E}) & \xrightarrow{\partial''} \Omega^2(\mathcal{E}).
\end{eqnarray} 
Since $\partial_A' \circ \partial_A' = 0$ and $\bar{\partial}^2 =0$, these are elliptic complexes, which generalize the notion of an elliptic operator. There is an analogous reduction to dual complexes for \eqref{dualcomplex}.
\end{remark}

\subsection{Constant Curvature}

At  the  end of Section \ref{OrbConn} we mentioned that a flat bundle could be transformed to a bundle of positive degree by twisting with a positive line bundle. To see how this is realized at the level of connections in our case, we take the unique holomorphic line bundle of degree 1 on $\mathbb{S}$, usually referred to as the {\it hyperplane bundle}. It has a constant curvature connection (associated to the Fubini-Study metric) given by
\begin{eqnarray} \label{FSConn}
        a^{(1,0)} = \frac{-\bar{z} dz}{1 + |z|^2}.
\end{eqnarray}


\section{Splitting the Monodromy}
\label{sec:Split}

 To further analyze the monodromy of $\partial_{A}'$ on $\Si$, we will make use of the following lemma. 
\begin{lemma} \cite{bib:hitch} \label{lem:Hitchen}
    Let $\rho$ be the monodromy representation associated to $\partial_{A}'$ with $\Gamma$ being  the abelian subgroup of this representation defined in Lemma \ref{Hitchen}.  It  is contained in a 1-parameter subgroup of SL(2, $\mathbb{C}$) of the form either 
    \begin{eqnarray*}
        \left\{ \left( \begin{array}{cc}
        \lambda & 0 \\ 0 & \lambda^{-1}
        \end{array}\right)\right\} &{or}& 
      \left\{ \pm \left( \begin{array}{cc}
        1 & \mu \\ 0 & 1
        \end{array}\right)\right\} . 
    \end{eqnarray*}
    The generators, $\rho_1, \dots, \rho_4$ of $\rho$ are of the form 
    \begin{eqnarray*}
         \left( \begin{array}{cc}
         0 & a \\ -a^{-1} & 0
        \end{array}\right) &{or}& 
       \pm \left( \begin{array}{cc}
        i & b \\ 0 & -i
        \end{array}\right) 
    \end{eqnarray*}
    in the respective previous cases related to $\Gamma$.
\end{lemma}
\noindent Going forward we will assume we are in the first (non-parabolic) case.
\smallskip

As a consequence of Theorem \ref{thm:NN} we have that the bundle $\hat{E}$ with connection $\hat{A}$ on $C$ that were defined in Section \ref{sec:HOrb} determine a unique holomorphic structure on $\hat{E}$. Going forward, unless stated otherwise, we will always take $\hat{E}$ to denote this holomorphic bundle.

\subsection{Invariant Sub-bundles} \label{invariant}
We now examine the structure of $\hat{E}$, building on an approach introduced in \cite{bib:hitch}.  In the neighborhood of the Weierstrass points,  $w_i$, the equation has an expansion of the form $z - z_i \sim y^2 + \cdots$. Since the residues of $\partial_{A}'$ near $z_i$ are $\pm 1/4$, this form of the branching shows that $\partial_{\hat{A}}'$ near a Weierstrass point has eigenvalues $\pm 1/2$. It then follows from direct local series expansion at a regular singular point that a covariant constant section, $s$, of $\partial_{\hat{A}}'$ has the form $s = y^{-1/2} (1,0)^\dagger + \cdots$ or $s = y^{1/2} (\alpha, \beta)^\dagger + \cdots$. 

Since we are assuming to be in the non-parabolic case of Lemma \ref{lem:Hitchen}, the monodromy has two distinct invariant subspaces in $\mathbb{C}^2$. Propagating these by continuation
(i.e., parallel transport) determines two sub-line bundles, $L_1$ and $L_2$, of the trivial bundle over $\hat{\Si}$. $\hat{A}$ induces a scalar connection on each of these. From the form of covariant constant sections of $\partial_{\hat{A}}'$ described above, it follows that $L_1$ extends to a line bundle on $C$ and the restricted scalar connection in a neighbourhood of $w_i$ extends to be meromorphic with  a simple pole of residue either -1/2 or +1/2.
This has the form of a first order Fuchsian scalar differential equation. It follows from \eqref{CWOrb} that $\hat{E}$ has parabolic degree 0.
Then, by \eqref{pardeg}, one has the following possibilities:
\begin{eqnarray*}
    0 = \deg L_1  \pm  1/2 \pm 1/2 \pm 1/2 \pm 1/2.
\end{eqnarray*}

But, $L_1$ is a holomorphic sub-bundle of the trivial bundle $\hat{E} \simeq C \times \mathbb{C}^2$. Hence the coordinate functions  of $\mathbb{C}^2$ restrict globally to $L_1$ to define two independent sections of the dual bundle $L_1^*$.  
By the Riemann-Roch theorem for an elliptic curve we have
\begin{eqnarray*}
    2 \leq \dim H^0(C, L_1^*) &=& 
     \deg L_1^* = - \deg L_1
\end{eqnarray*}
and so,
\begin{eqnarray*}
     \mp1/2 \mp 1/2 \mp 1/2 \mp 1/2 \leq \deg L_1 \leq -2.
\end{eqnarray*}
It follows that the only possibility is for the residues of the scalar connection on $L_1$ to all be $1/2$. It then further follows that $\deg L_1 = -2$. The same holds for $L_2$ by this same argument.

Now, the covering involution \eqref{invol} induces a nontrivial action on the bundles $L_1$ and $L_2$; i.e., it must interchange them with their respective fibers coinciding over the Weierstrass points $w_i$. 

\subsection{Constructing Holomorphic Sections}

Matrix solutions to Fuchsian equations  of the general form of 
\eqref{FuchsEqn} have been explictly constructed by methods of Riemann-Hilbert analysis. These solutions are built in terms of the function theory on $C$. We  will not actually need the explicit form of these solutions for our bifurcation analysis; however, it will be helpful to briefly describe it in order see more clearly how the various ingredients we have described combine to yield a solution. The explicit result was derived in \cite{bib:dikz, bib:eg} whose form for our case is stated below in
\eqref{szegoSoln}.
Most of the ingredients involved in this expression have been defined in Section \ref{inducedElliptic}.  The other ingredient is Riemann's theta function with characteristics,

\begin{eqnarray} \label{theta}
    \theta \left[\begin{array}{c}
       \delta  \\ 
       \epsilon \end{array}\right] (u; \tau) &=& \sum_{n  \in \mathbb{Z}}
       \exp(\pi i \tau (n +\delta)^2 + 2\pi i (z + \epsilon)(n + \delta) 
\end{eqnarray}
where $\delta, \epsilon \in \mathbb{R}$. These latter are referred to as {\it characteristics}. This is a holomorphic  Fourier  series on $\mathbb{C}$, automorphic with respect the period lattice $\Gamma$ with multiplicative factor of automorphy determined by the characteristics.
\smallskip 

A matrix solution, $Y(z)$,  in our case has entries $(s,r =  1.2)$  of the form
\begin{eqnarray} \label{szegoSoln}
  Y_{rs}(z) &=& X_{rs}(z)   
  \frac{\theta \left[\begin{array}{c}
       \delta  \\ 
       \epsilon \end{array}\right]
    \left(\int_{P_0^{(r)}}^{P^{(s)}} \phi ; \tau\right)}{\theta \left(\int_{P_0^{(r)}}^{P^{(s)}} \phi ; \tau\right)} 
    \frac{\theta \left(0  ; \tau\right)}{\theta \left[\begin{array}{c}
       \delta  \\ 
       \epsilon \end{array}\right]
    \left(0; \tau\right)}\\ \nonumber
    &&\mbox{where}\\ \nonumber
    X_{rs}(z) &=& \frac12\left[\left( (-1)^{s-r}\sqrt{\frac{p(z) q(z_0)}{p(z_0)q(z)}}\right)^{-1/2} + \left( (-1)^{s-r}\sqrt{\frac{p(z) q(z_0)}{p(z_0)q(z)}}\right)^{1/2}\right]\\ \nonumber
    q(z) &=& z - z_2 \\ \nonumber
    p(z) &=&  (z - z_1) (z- z_3). 
\end{eqnarray}
The characteristics, $\delta, \epsilon$, are determined by the monodromy of the connection with respect to $\Gamma$. 

$X$ solves the following canonical Riemann-Hilbert problem.  

\begin{enumerate}
\item $X(z)$ is analytic in $\mathbb{S} \backslash \mathcal{L}$, where  $\mathcal{L} = [z_1, z_2] \cup [z_3, \infty]$.
\item The limits, $X_\pm(z)$, of $X(z)$ as $z \to \mathcal{L}$ fromabove/below, satisfy the jump conditions
\begin{eqnarray*}
X_-(z) = X_+(z) \left(\begin{array}{cc} 0 & 1 \\ -1 & 0 \end{array}\right).
\end{eqnarray*}  
\item The behavior of $X(z)$ as $z\to \infty$ is specified in terms of data related to meromorphic differentials on $C$ \cite{bib:dikz, bib:eg}.
\end{enumerate}

We can make use of formula  \eqref{szegoSoln} to explicitly realize line bundles described in section \ref{invariant}. $Y$ is a fundamental solution of the pull-back of \eqref{FuchsEqn} to $C$. Parallel sections are then given by  $Y (\alpha, \beta)^\dagger$ for a vector $(\alpha, \beta)^\dagger$ in the fiber of $\hat{E}$ over $z_0$. Generically, due to monodromy, these sections are not well defined, even on $C$. However, if $(\alpha, \beta)^\dagger  = e_1 := (1,0)^\dagger$ or  $ e_2 :=(0,1)^\dagger$ (the column vectors of $Y$), the section is invariant up to a scalar multiple; i.e., it determines a sub-line bundle of $\hat{E}$ on $C$. These are  respectively the line bundles described in section \ref{invariant}.
\begin{proposition}
    The respective spans of the columns, $Y_1, Y_2$ of $Y$ are invariant under the monodromy action of $\Gamma$; i.e., they determine eigenvectors of $\Gamma$. The monodromy action of $\rho(\iota)$ interchanges these eigenvectors. 
\end{proposition}
\begin{proof}
  To study the action of $\Gamma$ we make use of the following automorphy relations for  the theta function \eqref{theta}.  
  \begin{eqnarray} \label{thetaa}
      \theta \left[\begin{array}{c}
       \delta  \\ 
       \epsilon \end{array}\right] (u  + 1; \tau) &=& e^{2\pi i \delta} \theta \left[\begin{array}{c}
       \delta  \\ 
       \epsilon \end{array}\right] (u; \tau) \\ \label{thetab}
       \theta \left[\begin{array}{c}
       \delta  \\ 
       \epsilon \end{array}\right] (u  + \tau; \tau) &=& e^{- 2\pi i \epsilon} 
       e^{- 2\pi i u} e^{- 2\pi i \tau}\theta \left[\begin{array}{c}
       \delta  \\ 
       \epsilon \end{array}\right] (u; \tau).   
  \end{eqnarray}
  Taking $a$ to be a loop encircling $z_1$ and $z_2$ and $b$ a loop going around $z_2$ and $z_3$ by crossing the cuts in $\mathcal{L}$, we can assess the monodromy of $Y_1 = Y (1,0)^\dagger$ and $Y_2 = Y (0,1)^\dagger$ using 
  \eqref{thetaa}, \eqref{thetab} and the jump condition for  $X$. From this it follows that along a path homologous to $a$, $Y_i$ transports $e_i$ to $(-1)^i e^{2\pi i \delta} e_i$ while along a path homologous to $b$,  $Y_i$ transports $e_i$ to $(-1)^i e^{2\pi i \epsilon} e_i$. This establishes the first statement of the proposition. For the second, we note that the jump matrix for $X$  interchanges $Y_1$ with $Y_2$. On the other hand, the action of $\iota$ on the entries of $Y$ also has the effect of interchanging $Y_1$ and $Y_2$. In the case of the theta functions this stems from the fact that the involution sends $\phi$
  to $-\phi$ and for $X$ because the expressions within parentheses can be rewritten as polynomials in  $z$ divided by $y$.
\end{proof}

It is clear from  this Proposition that the pair of solutions to \eqref{FuchsEqn} it describes are consistent with the notion of a YMH solution on an orbifold given in Definition \ref{YMH-orb}. This relation will be even clearer when we discuss solutions to the linearization  of the YMH equations. 
\bigskip

The following proposition will be used in later sections.

\begin{proposition}
   Let $\Psi$ be any vector projection of $Y$, the fundamental matrix solution \eqref{szegoSoln}. Then $\Psi \in \mathcal{H}^1$.
\end{proposition}
\begin{proof}
    We first note that $Y$ is holomorphic
    away from the  Weiersrass points $w_i$. This is evident for $X$. The only other non-holomorphic point arises at the zero of $\theta$ which occurs at $u = \frac12 + \frac{\tau}{2} $, which corresponds to a Weierstrass point.  Thus it suffices to study the behavior of $\Psi$ in the vicinity of such  a point. Since, $\Psi$ is a covariant constant section, it follows from the  local behavior of such sections near the the regular singular points, as described at the start of section \ref{invariant}, that $\Psi$ is locally $L^2$ and so by the compactness of $C$ it is globally $L^2$. Combining  this observation with Theorem \ref{biquard} and the fact, again from covariance constancy, that $\|\nabla_A \Psi\|^2 = 0$ it  follows that $\nabla \Psi$ is $L^2$-bounded. 
\end{proof}

\subsection{Twisting to Positive Constant Curvature}
We now turn to adding the scalar constant curvature connection \eqref{FSConn} to $A$ on $\Si$. In the holomorphic gauge the augmented covariant derivative becomes
\begin{eqnarray*}
    \partial'_A - \frac{\bar{z} dz}{1 +|z|^2}.
\end{eqnarray*}
Making the ansatz $\tilde{\Psi}(z) = f(z)Y(z)$ for scalar section $f$, and using the fact that $\partial'_A Y =0$, the parallel transport equations reduce to the scalar ODE
\begin{eqnarray*}
  \frac{d}{dz} f = \frac{\bar{z} f}{1 + |z|^2} 
\end{eqnarray*}
whose straightforward solution is $f = c(1 + |z|^2)$. 

\begin{proposition} \label{extFuchs}
The Fuchsian equation
$$\frac{d\tilde{\Psi}}{dz} =  \left(\sum_{k  = 1}^3 \frac{\hat{A}_k}{z - z_k} +  \frac{\bar{z}}{1 + |z|^2}\right)\tilde{\Psi}(z) $$  
where the residue matrices $\hat{A}_k$ have eigenvalues $\pm 1/2$, on $C$,
has fundamental matrix solution $\tilde{\Psi}=(1 + |z|^2)Y$ where $Y$ is specified by \eqref{szegoSoln}. The scalar connection $a^{(1,0)}$ is defined in \eqref{FSConn}. Let $L$ be the holomorphic line bundle of degree 1 on $\mathbb{S}$. The function $(1+ |z|^2)$ extends to a smooth section of $L$. 
Then $\tilde{\Psi}$ may be regarded as a section  of the bundle $End(E  \otimes L)$
pulled back to the universal cover of $\Sigma$. 
\end{proposition}

\begin{remark}
    We note that the pull-back of $L$ to $C$, which we denote by $\hat{L}$, has degree 2 so that $\deg \mbox{par}\,\, (L_i \otimes \hat{L})$ is zero. 
\end{remark}

 \begin{corollary} \label{solution}
     The section $(1 + |z|^2) Y_i$, is a section of $L_i \otimes \hat{L}$. It is the unique (up to scalar multiple) solution of $(\partial''_A + a^{(1,0)})\Psi = 0$, on the $\deg \mbox{par}\,\, 2$-bundle ($\hat{E} \otimes L$) on $C$ that is holomorphic on
     $\hat{\Sigma}$ and that equals $e_i$ at  $P^{(i)}_0$.
 \end{corollary}

\section{Linearization and Spectral Theory} \label{LinSpec}

 Equations \eqref{YMH-eqs} on the orbifold were posed with respect to the fixed conformally flat metric, $|du|^2$, introduced in section \ref{homotop}.  Going forward  we will work with a scale of metrics $h_r = r |du|^2$. The rescaled YMH equations have the form

 \begin{subequations} \label{rYMH-eqs}
    \begin{align}   
     & -\Delta_{A_{}}   \Psi  =  \frac{\kappa^2}{2} (r - |\Psi|^2)\Psi,
    \end{align}   \begin{align}
      &     d_{A}^* d_{A} A = J(\Psi, A), 
\end{align}
\end{subequations}

We consider the linearization of  these equations at the branch of normal solutions,  $(\Psi = 0, A^{b_r})$ where 
$A^{b_r} = \hat{A} + a^{(1,0)}$ with $b$ defined
in \eqref{extmagflux} and $b_r = b/r$.
This yields the system
\begin{eqnarray*} \label{LinSys}
    \left(-\Delta_{A^{b_r}} - \kappa^2 r\right) \psi  &=& 0\\
    d_{A^{b_r}}^*d_{A^{b_r}} \alpha &=& 0. 
\end{eqnarray*}
We make use of the previously noted fact that, as a consequence of Theorem \ref{biquard}, the sequence \eqref{complex} is  an elliptic complex associated to symbol $\partial'$. Hence  $-\Delta_{A^{b_r}}$ is an elliptic operator acting on sections of $\hat{E} \otimes L$ over $C$ \cite{bib:wells}.
One then has the following results about this magnetic Laplacian.  

 \begin{proposition}\label{prop:lin-probl-abstr} 
  On $(C, h_r)$  one has that

\begin{enumerate} 
\item   the spectrum of $-\Delta_{A^{b_r}}$ is purely discrete;   

\item the operator $ -\Delta_{A^{b_r}}$ satisfies the inequality 
\[\lan -\Delta_{A^{b_r}}  \xi, \xi\ran \ge b_r \|\xi\|^2 \] 
for $\xi \in \mathcal{H}^1$;
\item $b_r$ is an eigenvalue 
of the operator $-\Delta_{A^{b_r}}$ with an eigenfunction $\xi$ iff $\xi$ solves the equation
   \begin{equation} \label{1st-ord-eq} \p_{A^{b_r} }'\xi=0. 
   \end{equation} 
   \end{enumerate} 
   \end{proposition}
\begin{proof}
Item 1 follows from ellipticity and the fact that $C$ is compact. For item 2 one makes use of the Weitzenbock representation 
\begin{eqnarray} \nonumber
   \left( \p_{A^{b_r} }'^*\p_{A^{b_r} }'\right) \xi &=& \frac12 \left( -\Delta_{A^{b_r}} - i*F_{A^{b_r}} \right)\xi\\ \label{Weitz}
   &=& \frac12 \left( -\Delta_{A^{b_r}} - {b_r} \right)\xi
\end{eqnarray}
which follows by direct calculation, making use of the K\"ahler identity
\begin{eqnarray*}
    d^*d = 2 \partial^* \partial
\end{eqnarray*}
\cite{bib:wells}.
    It then follows from \eqref{Weitz}  that 
    \[ \left\langle \left( -\Delta_{A^{b_r}} - {b_r} \right)\xi, \xi\right\rangle = 
    \langle  \p_{A^{b_r} }'^*\p_{A^{b_r} }' \xi, \xi \rangle = \langle  \p_{A^{b_r} }' \xi, \p_{A^{b_r} }'\xi \rangle \ge 0.\]
    Finally item 3 is a direct consequence of 
    the last equation.
  \end{proof} 
  \noindent Taking $\xi$ to be the pullback of the section-pair $(1+|z|^2) Y_i$ specified in Corollary \ref{solution} we see that
  \begin{corollary} \label{cor:main}
      $b_r$ is an eigenvalue 
of $-\Delta_{A^{b_r}}$  with non-degenerate (i.e., one-dimensional) eigenspace.
  \end{corollary}

\section{Existence and Uniqueness of YMH solutions} \label{Existence}

We are now going to establish the existence of solutions to YMH on our orbifold $\Si$ with non-trivial $\Psi$ by outlining a bifurcation analysis in the vicinity of a normal solution. In this we follow the approach taken in \cite{bib:cers, bib:esz}.

Define a nonlinear map
\begin{eqnarray}  \label{Func}
    \mathcal{F} &:& \mathcal{H}^s \times \vec{\mathcal{H}}^s \times \mathbb{R} \to \mathcal{H}^{s-2} \times \vec{\mathcal{H}}^{s-2} \\ \nonumber
    && (\Psi, a, r) \mapsto \left( -\Delta_{A^{b_r} + a} \Psi + \kappa^2 (|\Psi|^2 - r)\Psi ,P( d_{A^{b_r}}^*d_{A^{b_r}} a -  J(\Psi, a))\right)
\end{eqnarray}
where  $\mathcal{H}^s$ and $\vec{\mathcal{H}}^s$ are as in Definition \ref{sobolev}, 
\begin{eqnarray*}
 J(\Psi, a) &:=& \mbox{Im} \left(\bar{\Psi} \nabla_{A^{b_r} + a} \Psi\right) \mbox{and}\\
 P: \vec{\mathcal{H}}^{s-2} &\to& \vec{\mathcal{H}}^{s-2} 
\end{eqnarray*}
is the projection onto the space of $d_{A^{b_r}}$ - co-closed matrix valued 1-forms. 

We first consder the situation in which we are working on a  smooth, closed Riemann surface in which case we may take $s \geq 2$ so that everything is manifestly well-defined. The  projection $P$, which maps to a space that is orthogonal to the subspace of gauge symmetry null modes, is introduced to ensure one has sufficient coercivity to guarantee existence of minimizers. But one also has the following result which is a direct extension of a result from the Ginzburg-Landau case.
\begin{proposition} \cite{bib:cers}
    Assume $(\Psi, A^{b_r})$ solves \eqref{YMH-eqs-psi}. Then $(d_{A^{b_r}}^*d_{A^{b_r}} a - J(\Psi, a))$  is a 
    $d_{A^{b_r}}$ - co-closed one-form (i.e., $d_{A^{b_r}}^* (d_{A^{b_r}}^*d_{A^{b_r}} a - J(\Psi, a))  = 0$).
\end{proposition}
\noindent Hence one has that

\begin{corollary} \cite{bib:cers}
A solution $(\Psi, a, r)$  to $\mathcal{F}=0$ also solves \eqref{rYMH-eqs}.
\end{corollary}

In the case of interest to us, $s = 1$ for which the image of $\mathcal{F}$ lands in 
$\mathcal{H}^{-1} \times \vec{\mathcal{H}}^{-1}$. This is to be understood as a space of tempered distributions on $\mathcal{H}^{1} \times \vec{\mathcal{H}}^{1}$.
But this is precisely the sense in which the equations \eqref{rYMH-eqs} are to be understood as described in Proposition \ref{distribution}. Implementing the coercivity condition in this generalization means that one replaces the second component of $\mathcal{F}=0$ by the requirement that
\begin{eqnarray*}
    (d_{A^{b_r}} a, d_{A^{b_r}} \alpha) - (J(\Psi, a), \alpha) = 0
\end{eqnarray*}
for all Schwarz class $\alpha$ in the orthogonal complement of $d_{A^{b_r}}$-closed one-forms.

\subsection{Lyapunov-Schmidt Bifurcation Analysis}
We now consider the linearization  of $\mathcal{F}$ along the trivial branch of normal solutions $(\Psi, \alpha, r)= (0,0,r)$
\begin{eqnarray} \label{Lin}
    dF(0,0,r) &=& \diag \left( -\Delta_{A^{b_r}} - \kappa^2 r, d_{A^{b_r}}^* \circ  d_{A^{b_r}}\right): X^s \to X^{s-2}.
\end{eqnarray}
Defining $N_r := \mbox{Null} \,\,dF(0,0,r)$, one sees from \eqref{Lin} that 
\begin{eqnarray*}
    N &\equiv& N_{b/\kappa^2} = K \times \Omega\\
    K &=& \mbox{Null}\,\, (-\Delta_{A^{b_r}} - b)\\
    \Omega &=& \mbox{Null} \,\,  d_{A^{b_r}}^* \circ  d_{A^{b_r}} = \left\{   d_{A^{b_r}} - \mbox{harmonic 1-forms on $C$}\right\}.
\end{eqnarray*}
When $s = 1$, this is understood in the weak sense described earlier. Note that, by ellipticity, weakly harmonic 1-forms are harmonic. 

\begin{theorem}
    Let $(\Sigma, h_r), r>0$ be an orbifold sphere with finite area metric induced by   $r |du|^2$ on $C$. Let $E \to \Sigma$  be a rank 2 unitary vector bundle with connection of constant curvature $b = \frac{2 \pi \deg {par} E}{2 |\Sigma| } > 0$ where $|\Sigma|$ denotes the total area of $\Sigma$ when $r  = 1$. Suppose that $r$ satisfies $0 <  |\kappa^2  r -  b| << 1$ and $(\kappa - \sqrt{b/r})(\kappa - \kappa_c(r))  > 0$ where $\kappa_c(r) := \sqrt{\frac12 \left( 1 - \frac1{\beta(r)}\right)}$ and 
    $\beta(r):= \min \{ \langle  |\xi|^4\rangle: \xi \in K, \langle  |\xi|^2\rangle = 1\}$. $(\langle f \rangle = \frac1{{r|\Sigma|}}\int  f) $. 
    
    Then, for each  $r$ satisfying these conditions, there exists a solution $(\Psi(r), A(r))$, in a neighborhood $U \subset \mathcal{H}^{-1} \times \vec{\mathcal{H}}^{-1}$ of  $(0, A^{b_r}$$ )$ satisfying the scaled YMH equations \eqref{rYMH-eqs} in the sense of Definition \ref{YMH-orb}. This solution  is unique in $U$ up to gauge symmetry. 
\end{theorem}

\begin{proof}
  The goal is to show that a non-trivial branch of solutions bifurcates from $N$ if $0 < |\kappa^2 r - b| << 1.$
Let $Q$ denote projection within $\vec{\mathcal{H}}^{-1}$ onto $N$ and $Q^\perp =1 - Q$ projection onto $N^\perp$. Let $u$ denote a pair $(\Psi, \alpha)$ in $\vec{\mathcal{H}}^{-1}$. Set $v = Qu$ and $w = Q^\perp u$. Split $\mathcal{F}=0$ into 
\begin{eqnarray} \label{Q-eq}
   Q\mathcal{F}(v+w, r) &=& 0\\ \label{QPerp-eq}
   Q^\perp \mathcal{F}(v+w,r)  &=& 0. 
\end{eqnarray}

Since $d_w Q^\perp|_{w=0}$ is invertible on Ran $Q^\perp$ it follows that for $(v,r)$ in $N \times \mathbb{R}_{\geq 0}$ and $|\kappa^2 r - b|$ both sufficiently small, the implicit function theorem provides a unique solution $w(v,r)$ to \eqref{QPerp-eq}. 
This reduces the problem to solving the bifurcation equation $Q\mathcal{F}(v + w(v,r), r) = 0$ for $v$. We saw in Section \ref{sec:Split} that $K$ is one dimensional. By ellipticity $\Omega$, the space of $d_{\hat{A}^b} $ -harmonic 1-forms on $C$, is finite dimensional. Hence, choosing bases, this bifurcation equation becomes a finite dimensional system of ODEs. Establishing the solvability of this system then proceeds exactly as was done in \cite{bib:cers, bib:esz}. 

\end{proof}
\section{Conclusions}

As we stated at the outset,  the purpose of this article was to provide a proof of the concept that the methods used to establish the existence of a superconducting lattice states within the Ginzburg-Landau model could be extended to finding non-trivial Higgs field states on Riemann surfaces within the non-abelian Yang-Mills-Higgs models. This was carried out for the special case when the surface is a spherical orbifold with four conical singularities. It is natural to extend this analysis to spherical orbifolds with $n > 4$ conical singularities. The main difference  here is that uniformization will now take place on the upper half plane $\mathbb{H}$ with respect to a conformally flat metric but of constant curvature -1. However extending this to surfaces of genus $>0$ (even in the smooth, non-orbifold, case) will present challenges for the construction of holomorphic sections for bundles of rank $>1$ due to the increased complexity of the monodromy representations.  

In another direction, the stability analysis for the orbifold Higgs field found here should be tractable. In \cite{bib:cers}, Corllary 1.3, we stated precise conditions, in terms of $b_r$ for the stability of non-normal solutions to GLE found there.  It is reasonable to expect a similar result for the case considered here. One may also consider stability with respect to location of the conical singularities (equivalently with respect to the lattice $\Gamma$). This is the analogue of the original Abrikosov question to determined which lattice geometry is energetically preferred.  It would  be natural to attack this by first varying within the class of isomonodromic connections which is directly related to the Painlev\'e VI equation.
We leave this for future consideration.


\begin{thebibliography}{}

\bibitem[Abri]{bib:abri}
A.A. Abrikosov, {\em On the magnetic properties of superconductors of the second group}, J. Explt.  Theoret. Phys. (USSR) {\bf 32}, 1147 - 1182, (1957).

\bibitem[AL]{bib:al}
D. Arinkin and S. Lysenko, {\em On the moduli of $SL(2)$-bundles with Connections on $P^1\backslash \{ x_1, \dots x_4\}$} {\bf  19}, 983 - 999, (1997). 

\bibitem[BGT]{bib:bgt}
E. Barany, M. Golubitsky, and J. Turski, {\em Bifurcations with local gauge symmetries in the Ginzburg-Landau equations}, Phys. D {\bf 56}, 36–56, (1992).

\bibitem[Bau]{bib:bau}
S. Bauer, {\em Parabolic bundles, elliptic surfaces and $SU(2)$-representation spaces of genus zero Fuchsian groups}, Math. Ann. {\bf 29},  509-526, (1991) .

\bibitem[Biq]{bib:biq}
O. Biquard, {\em Fibr\'es paraboliques stable et connexions singuli\`eres plates},  
{\em Bull. Soc.math. France {\bf 119}}, 231 - 257, (1991).

\bibitem[Biq1]{bib:biq1}
O. Biquard, {\em Prolongement d'un fibre holomorphe hermitien a courbure $L^p$  sur une courbe ouverte}, International Journal of Mathematrics {\bf 3}, 441 - 453, (1992)

\bibitem[Both]{bib:both}
T. Bothner, {\em On the origins of Riemann -Hilbert problems in mathematics}, Nonlinearity {\bf 34}, R1, (2021). 

\bibitem[CERS]{bib:cers} 
D.Chouchkov,  N.M.Ercolani, S.Rayan and I.M. Sigal, {\em Ginzburg-Landau equations on  Riemann surfaces of higher genus},  Ann. I. H. Poincar\'e - AN {\bf 37}, 79 - 103, (2020).

\bibitem[CJRF]{bib:cjrf} D. Cribier, B. Jacrot, L.M. Rao and B. Farnaux, {\em Mise en evidence par diffraction de neutrons d'une structure periodique du champ magnetique dans le niobium superconducteur}, Phys Lett  {\bf 9}, 106,  (1964).

\bibitem[DIKZ]{bib:dikz}
P. A. Deift, A. R. Its, A. Kapaev, and X. Zhou, {\em On the algebro-geometric integration of the Schlesinger equations}, Commun. Math. Phys. {\bf 203}, 613–633, (1999).

\bibitem[EG]{bib:eg}
V.Z. Enolski and T. Grava, {\em Singular $Z_N$-curves and the Riemann-Hilbert problem}, IMRN {\bf 32},  1619 - 1683, (2004).

\bibitem[ESZ]{bib:esz}
N.M. Ercolani, I.M.Sigal and J. Zhang, {\em Ginzburg-Landau equations on  non-compact Riemann surfaces}, Journal of Functional Analysis {\bf 285}, 110074, (2023).

\bibitem[ET]{bib:et} U. Essmann and H. Tr\"auble, {\em  The direct observation of individual flux lines in type II Superconductors}, Phys. Lett. A {\bf 24}, 526 - 527, (1967).

\bibitem[FK]{bib:fk}
H. Farkas and I. Kra, {\em Riemann Surfaces} Graduate Texts in Mathematics {\bf 71},
Springer-Verlag,Berlin, Heidelberg, New York (1980)
\bibitem[Feyn]{bib:feyn}
R.P. Feynman, R. B. Leighton and M. Sands  {\em The Feynman Lectuures on Physics}, Volume III, Pearson, (2012).

\bibitem[Hitch]{bib:hitch}
N.J. Hitchen, {\em Twistor spaces, Einstein metrics and  isomonodromic deformations},
J. Differential Geometry {\bf 42}, 30 -  112  (1995).

\bibitem[Koby]{bib:koby}
S. Kobayashi,  {\em Differential Geometry of Complex Vector Bundles}, Publications of the Mathematical Society of Japan {\bf 15}  Princeton  University Press and Iwanami Shoten  (1987).  

\bibitem[Kor]{bib:kor}
D. Korotkin, {\em Solution of matrix Riemann-Hilbert problems with quasi-permutation monodromy matrices}, Mathematische Annalen {\bf 329}, 335 - 364, (2004).

\bibitem[Lor]{bib:lor}
F. Loray, {\em Introduction to moduli spaces of connections: some explicit constructions},  {\em hal-01133705} (2015).

\bibitem[MS]{bib:ms}
V. B. Mehta and C. S. Seshadri, {\em Moduli of vector bundles on
curves with parabolic structures}, Math. Ann. {\bf 248}, 205-239, (1980).

\bibitem[NS]{bib:ns}
B. Nasatyr and B. Steer, {\em  Orbifold Riemann surfaces and the  Yang-Mills-Higgs equations}, Ann. S. Norm. Sup. Pisa Ce. Sci. (4)  {\bf 22}, 596-643, (1995).

\bibitem[Odeh]{bib:odeh}
F. Odeh, {\em Existence and bifurcation theorems for the Ginzburg-Landau equations},
J. Math. Phys. {\bf 8}, 2351–2356 (1967).

\bibitem[Sesh]{bib:sesh}
C. S. Seshadri, {\em Moduli of vector bundles with parabolic structures},
Bull. Am. Math. Soc. {\bf 83} 124-126, (1977).

\bibitem[Tak]{bib:tak}
P. Tak\'a$\check{c}$, {\em Bifurcations and vortex formation in the Ginzburg-Landau equations}, Z. Angew. Math. Mech. {\bf  81}, 523–539 (2001).

\bibitem[Troy]{bib:troy}
M. Troyanov, {\em Prescribing curvature on compact surfaces with  conical singularities}, {\em Transactions of the AMS}  {\bf 324}, 793 -  821, (1991).

\bibitem[TS]{bib:ts}
T. Tzaneteas and I.M. Sigal, {\em  Abrikosov lattice solutions of the Ginzburg–Landau equations}, Contemp. Math. {\bf 535}, 195–213 (2011). 

\bibitem[Wells]{bib:wells}
R.O. Wells, Jr., {\em Differential Analysis on Complex Manifolds, Third Edition}, 
Graduate Texts in Mathematics {\bf 65}, Springer-Verlag,Berlin, Heidelberg, New York,  (2008).

\end{thebibliography}
\end{document}